\RequirePackage{fix-cm}
\documentclass{svjour3}                     
\smartqed  
\usepackage{graphicx}
\usepackage{latexsym}
\usepackage{amssymb}
\usepackage{amsmath}
\usepackage{enumerate}
\newtheorem{fact}{Fact}
\journalname{Cryptography and Communications}
\begin{document}

\title{The Fourier Spectral Characterization for the Correlation-Immune Functions over $\mathbb{F}_p$ \thanks{Z. Wang was supported in part by NSFC (No. 61671013, 61672410), NSF of Shaanxi Province (No. 2018JM6076), and the Programme of Introducing Talents of Discipline to Universities (China 111 Project, No. B16037)}
}


\author{Zilong Wang         \and
        Jinjin Chai
        \and
        Guang Gong
}


\institute{Zilong Wang and Jinjin Chai\at
              State Key Laboratory of Integrated Service Networks\\
              School of Cyber Engineering, Xidian University \\
              \email{zlwang@xidian.edu.cn,\ jj\_chai@163.com}           
           \and
           Guang Gong \at
              Department of Electrical and Computer Engineering\\
              University of Waterloo\\
              \email{ggong@uwaterloo.ca}\\
}

\date{Received: date / Accepted: date}

\maketitle

\begin{abstract}
The correlation-immune functions serve as an important metric for measuring resistance of a cryptosystem  against correlation attacks. Existing literature emphasize on matrices, orthogonal arrays and Walsh-Hadamard spectra to characterize the correlation-immune functions over $\mathbb{F}_p$ ($p \geq 2$ is a prime). 
Recently, Wang and Gong investigated  the Fourier spectral characterization over the complex field for correlation-immune Boolean functions. In this paper, the discrete Fourier transform (DFT) of non-binary functions was studied. It was shown that a function $f$ over $\mathbb{F}_p$ is $m$th-order correlation-immune if and only if its Fourier spectrum vanishes at a specific location under any permutation of variables. Moreover, if $f$ is a symmetric function, $f$ is correlation-immune if and only if its Fourier spectrum vanishes at only one location.

\end{abstract}

\section{Introduction}
\label{intro}
Correlation immunity is one of the important properties for cryptographic functions. It was introduced to prevent some cryptosystems from correlation attacks~\cite{siegenthaler1985,Rueppel1986Analysis}(or {\lq divide and conquer attack \rq} proposed by Siegenthaler~\cite{Siegenthaler1984}).
Moreover, correlation-immune functions are closely related to secret-sharing schemes and error-correcting codes~\cite{gopalakrishnan1996applications,Bierbrauer96,wu1996construction,ding1991stability}. The constructions of correlation-immune functions with desired nonlinearity, algebraic immunity, propagation and strict avalanche criteria, such as \cite{Carlet2010Boolean,carlet2018constructing,su2014construction,tang2014class,zhang2014generalized}, just list a few here, were well studied.

There are several methods to characterize correlation-immune Boolean functions. In 1989, Xiao and Massey~\cite{Xiao1988A} showed that a Boolean function is $m$th-order correlation-immune if and only if its Walsh-Hadamard transform vanishes for all points with Hamming weights between $1$ and $m$. A similar concept, called invariants of Boolean functions, which measure the distances between a Boolean function and all affine functions,  was introduced by Golomb~\cite{Golomb1959On}\cite{Golomb1967Shift} in 1959. In fact, invariants are the same concept of Walsh-Hadamard spectral characterization of correlation-immune functions. But the motivation to propose the invariants of Boolean functions was not explicitly mentioned until 1999 \cite{Golomb1999On}. In addition to Walsh-Hadamard spectral characterization, other methods to characterize correlation-immune Boolean functions, such as matrices~\cite{Gopalakrishnan1995}, orthogonal arrays~\cite{Camion1991On,Bierbrauer96}, and Fourier spectra~\cite{WangDiscrete} were proposed.

For non-binary functions, the characterization of correlation immunity was also studied.
Zhang and Xiao~\cite{Zhang1995Spectral} proved that a function $f$ over $\mathbb{F}_p$ for $p$ prime is $m$th-order correlation-immune if and only if its Chrestenson cyclic spectrum vanishes for all points with weights between $1$ and $m$ by generalizing Walsh-Hadamard spectral characterization. Gopalakrishhan and Stinson~\cite{Gopalakrishnan1995} investigated the correlation-immune functions over the finite fields and derived their matrices, spectra, and orthogonal arrays characterizations. Feng~\cite{Feng1999Three} obtained two necessary and sufficient conditions for the correlation-immune functions over $\mathbb{Z}_N$ in terms of Chrestenson linear spectrum and Chrestenson cyclic spectrum. It is clear that the methods to characterize the correlation-immune functions, including Walsh-Hadamard spectra, matrices, orthogonal arrays, were all generalized from the Boolean functions to non-binary functions.

Recently, Wang and Gong~\cite{WangDiscrete} investigated discrete Fourier transform over the complex field of Boolean function $f(\mathbf{x})$, and deduced an equivalent condition for an $m$th-order  correlation-immune Boolean function, that is, the Fourier spectrum of $f(\mathbf{x})$ under any permutation of $n$ variables vanishes at a particular location. Inspired by this result, we try to characterize $m$th-order correlation-immune functions over $\mathbb{F}_p$ for $p$ prime by discrete Fourier transform over the complex field. However, the method used in~\cite{WangDiscrete} cannot be extended to prove the result here. In this paper, we used a different approach than the literature~\cite{WangDiscrete}, and obtained that a function $f$ over $\mathbb{F}_p$ is $m$-order correlation-immune if and only if its Fourier spectrum vanishes at a specific location under any permutation of variables. In particular, a symmetric function $f$ is correlation-immune if and only if its Fourier spectrum vanishes at a specific location. The Fourier spectral characterization for Boolean functions in \cite{WangDiscrete} can be regarded as a special case of the results in this paper for $p=2$.

The rest of this paper is organized as follows. In Section~\ref{sec:1}, we introduce the definitions of the correlation immunity, resiliency and discrete Fourier transform over the complex field of the functions, review previous results on the characterizations of correlation-immune functions over $\mathbb{F}_p$ as well. In Section~\ref{section:s4}, we present the main result on Fourier spectral characterization of correlation-immune functions. In Section~\ref{section:s5}, we show an equivalent condition for $m$-order resilient functions. Section~\ref{section:s6} concludes the paper.

\section{Preliminaries}\label{sec:1}

Let $f$ : $\mathbb{F}^n_p\rightarrow \mathbb{F}_p$ be a function with variable $\mathbf{x}=(x_1,x_2,\cdots ,x_n)$, where $\mathbb{F}_p$ denotes Galois Field with $p$ elements for prime $p$, and $\mathbb{F}^n_p=\{(x_1,x_2,\cdots ,x_n)|x_i \in \mathbb{F}_p,1 \leq i \leq n \}$ represents an $n$-dimensional vector space over $\mathbb{F}_p$.

\subsection{Correlation Immunity and Resiliency}\label{subsec:1}

\begin{definition}\label{definition:d1}
Let $X_1,X_2,\cdots ,X_n$ be independent and uniformly distributed random variables (that is, for any $a \in \mathbb{F}_p$, assume that $P(X_i=a)=p^{-1}$, $1 \leq i \leq n$). A function $f$ is said to be {\em $m$th-order correlation-immune} if the random variable $Z=f(X_1,X_2,\cdots ,X_n)$ is statistically independent of every $m$-subset of random variables chosen from $X_1,X_2,\cdots ,X_n$, which means $$P_r(Z=t|X_{i_j}=a_j,1\leq j\leq m)=P_r(Z=t),$$ for every $m$-subset $\{i_1,\cdots ,i_m\}\subseteq \{1,\cdots ,n\}$, $a_j\in\mathbb{F}_p(1\leq j\leq m)$, and $t\in \mathbb{F}_p$.
\end{definition}
A function $f$ is said to be balanced if every possible output occurs with equal probability $p^{-1}$. Furthermore, if $f$ is $m$th-order correlation-immune and balanced, then $f$ is said to be {\em $m$-resilient}.

\subsection{Associated Polynomial}\label{subsec:2}

We describe a sequence $\mathbf{f}$ of length $p^n$ corresponding to the function $f$ by listing the values taken by $f(x_1,x_2,\cdots ,x_n)$ as $(x_1,x_2,\cdots ,x_n)$ which ranges over all its $p^n$ values in lexicographic order. In other words, sequence $\mathbf{f}$ is defined by
\begin{equation}\label{f(k)}
  \mathbf{f}=(f(0),f(1),\cdots ,f(p^n-1)),
\end{equation}
where $f(k)=f(x_1,x_2,\cdots ,x_n)$ and $(x_1,x_2,\cdots ,x_n)$ is the $p$-adic representation of the integer $k$  for $0\leq k \leq p^n-1$, i.e., $k=\sum_{i=1}^n x_i p^{i-1}$.

Let $\omega=exp(\frac{2\pi\sqrt{-1}}{p})$ be a $p$th primitive root of unity over the complex field. The polynomial associated with the function $f(x_1,x_2,\cdots ,x_n)$ is defined by the $Z$-transform of the sequence $\omega^{\mathbf{f}}=(\omega^{f(0)},\omega^{f(1)},\cdots ,\omega^{f(p^n-1)})$, i.e.,
\begin{equation}\label{Z-transform}
F(z)=\sum^{p^n-1}_{k=0}\omega^{f(k)}z^k.
\end{equation}

\subsection{Discrete Fourier Transform}\label{subsec:3}

We now introduce the concept of the discrete Fourier transform (DFT) over the complex field of the function $f(x_1,x_2,\cdots ,x_n)$, which actually is the DFT of the sequence described by function $f$. Note that DFT over the complex field introduced here is the traditional DFT, which is different from the DFT over the finite field \cite{Golomb2005Signal}.
\begin{definition}
Let $\xi=exp(\frac{2\pi\sqrt{-1}}{N})(N=p^n)$ be an $N$th primitive root of unity over the complex field. The {\em discrete Fourier transform } (DFT) of function $f(x_1,x_2,\cdots ,x_n)$ over the complex field is defined by
\begin{equation}\label{DFT}
\mathcal{F}_f(j) =\sum^{N-1}_{k=0}{\omega}^{f(k)}\xi^{-kj},0\leq j \leq N-1.
\end{equation}
\end{definition}
Then the inverse discrete Fourier transform (IDFT) of the function $f(x_1,x_2,\cdots ,x_n)$ is given by
\[\omega^{f(k)}=\frac{1}{N}\sum_{j=0}^{N-1}\mathcal{F}_f(j)\xi^{kj},0 \leq k \leq N-1.\]

Let $F_N=(f_{kj})$ be an $N\times N$ matrix whose entries are defined by $f_{kj}=\xi^{-kj}$, $0\le k, j \le N-1$, i.e.,
\[
F_N=
\begin{pmatrix}
1 & 1 & 1 & \cdots & 1\\
1 & \xi^{-1} & \xi^{-2} & \cdots & \xi^{-(N-1)}\\
1 & \xi^{-2} & \xi^{-4} & \cdots & \xi^{-2(N-1)}\\
\vdots &\vdots & \vdots & \vdots & \vdots \\
1 & \xi^{-(N-1)} & \xi^{-2(N-1)} & \cdots & \xi^{-1}\\
\end{pmatrix}.
\]
This is called the {\em DFT matrix}. Under this notation, we have DFT and IDFT respectively:\\
\begin{equation*}
\begin{pmatrix}
	\mathcal{F}_f(0)\\
	\mathcal{F}_f(1)\\
	\vdots\\
	\mathcal{F}_f(N-1)\\
\end{pmatrix}
	 = F_N
\begin{pmatrix}
	\omega^{f(0)}\\
	\omega^{f(1)}\\
	\vdots\\
	\omega^{f(N-1)}\\
\end{pmatrix}
and
\begin{pmatrix}
	\omega^{f(0)}\\
	\omega^{f(1)}\\
	\vdots\\
	\omega^{f(N-1)}\\
\end{pmatrix}
	 =\frac{1}{N}F_N^*
\begin{pmatrix}
	\mathcal{F}_f(0)\\
	\mathcal{F}_f(1)\\
	\vdots\\
	\mathcal{F}_f(N-1)\\
\end{pmatrix}
\end{equation*}	
where $F_N^*$ is the Hermitian transpose of $F_N$, and the entries $f^*_{kj}=\xi^{kj}$, $0\le k, j \le N-1$.

The relationship of DFT spectra of a sequence and its auto-correlation function has been studied before.
\begin{definition}
Let $\mathbf{f}$  be the sequence described by function  $f(\mathbf{x})$: $\mathbb{F}^n_p\rightarrow \mathbb{F}_p$. The {\em auto-correlation} function of $\mathbf{f}$ is defined by 
\[C_{f}(t)=\sum_{k=0}^{N-1}\omega^{f(k+t)-f(k)},0\le t<N,\]
where  $f(k)$ is defined by (\ref{f(k)}) and the addition in $(k+t)$ is over $\mathbb{Z}_N$.
\end{definition}
It is known that the squared magnitude of the DFT spectra of $f(\mathbf{x})$ and the autocorrelation of sequence described by function $f(\mathbf{x})$ are a DFT pair, i.e.,
 \[|\mathcal{F}_f(j)|^2 = \sum_{t=0}^{N-1} C_f(t) \xi^{-tj}, \,\,0\le j<N \]
     and\[C_f(t)= \frac{1}{N}\sum_{j=0}^{N-1} |\mathcal{F}_f(j)|^2 \xi^{jt}, 0\le t<N. \]

Recall the definition of  the DFT and associated polynomial of the function $f(x_1,x_2,\cdots ,x_n)$, it is obvious that $\mathcal{F}_f(j)=F(z=\xi^{-j})$. We shall use polynomial $F(z)$ to explore the DFT of function $f$ in the rest of the paper.

\subsection{Previous Results on the Characterizations of Correlation-Immunity}\label{subsec:4}

Correlation-immune and resilient functions were initially defined by the probabilistic terminology (Definition \ref{definition:d1}), but it is difficult to determine the correlation-immune functions by probabilistic method. Therefore, some algebraic and combinatorial methods were proposed to characterize the correlation-immune functions. For $p=2$, these methods to study Boolean functions are refer to Walsh spectra~\cite{Xiao1988A}, matrices~\cite{Gopalakrishnan1995}, orthogonal arrays~\cite{Camion1991On}, and Fourier spectra~\cite{WangDiscrete}. For general prime $p$, the above first three characterizations of the  correlation-immune functions were generalized to Chrestenson spectra~\cite{Feng1999Three}, matrices~\cite{Gopalakrishnan1995}, and orthogonal arrays~\cite{Gopalakrishnan1995}, respectively.

Chrestenson transform, which is also referred to generalized Walsh-Hadamard transform in the literature, was introduced in \cite{Zhang1995Spectral,Feng1999Three} to study the correlation-immunity of the non-binary functions.
For a vector $\mathbf{c}=(c_1,c_2,\cdots ,c_n)\in \mathbb{F}^n_p$, let $wt(\mathbf{c})$ denotes the Hamming weight of $\mathbf{c}$.

\begin{definition}
{\em Chrestenson linear spectrum} of $f(\mathbf{x})$ is defined by
\[S_f(\mathbf{c})=\frac{1}{p^n}\sum_{x\in\mathbb{F}^n_p}f(\mathbf{x})\omega^{\mathbf{c}\cdot\mathbf{x}},\]
and {\em Chrestenson cyclic spectrum} of $f(\mathbf{x})$ is given by
\[S_{(f)}(\mathbf{c})=\frac{1}{p^n}\sum_{x\in\mathbb{F}^n_p}\omega^{f(\mathbf{x})-\mathbf{c}\cdot\mathbf{x}},\]
respectively, where $\omega=exp(\frac{2\pi\sqrt{-1}}{p})$, $\mathbf{c}\cdot\mathbf{x}=c_1x_1+c_2x_2+\cdots +c_nx_n$ is the inner product of $\mathbf{c}$ and $\mathbf{x}$.
\end{definition}

Correlation-immune functions are characterized by their Chrestenson spectra.
\begin{fact}
(\cite{Feng1999Three}) $f$ is $m$th-order correlation-immune $\Longleftrightarrow$ $S_{f+a}(\mathbf{c})=0$ for  $\forall a\in\mathbb{F}_p$ and $\forall \mathbf{c}\in \mathbb{F}^n_p$ with $1\leq wt(\mathbf{c})\leq m$.
\end{fact}

\begin{fact}
(\cite{Zhang1995Spectral},\cite{Feng1999Three}) $f$ is $m$th-order correlation-immune $\Longleftrightarrow$ $S_{(f)}(\mathbf{c})=0$ for $\forall \mathbf{c}\in \mathbb{F}^n_p$ with $1\leq wt(\mathbf{c})\leq m$.
\end{fact}

For $\forall \mathbf{c}\in \mathbb{F}^n_p$, define a $p\times p$  matrix $N=N(\mathbf{c})=(\eta_{ij})$, where
\[\eta_{ij}=p^n \cdot P_r(\mathbf{x}\cdot\mathbf{c}=i~and~f(\mathbf{x})=j),i,j\in \mathbb{F}_p.\]
Correlation-immune functions are characterized by the following matrix method.
\begin{fact}
(\cite{Gopalakrishnan1995}) $f$ is $m$th-order correlation-immune $\Longleftrightarrow$ the rows of the matrix $N(\mathbf{c})$ are all identical for $\forall \mathbf{c}\in \mathbb{F}^n_p$ with $1\leq wt(\mathbf{c})\leq m$.
\end{fact}

An $M\times n$ matrix $A$ with entries from a set of $p$ elements is an {\em orthogonal array} if any set of $m$ columns of $A$ contains all $p^m$ possible row vectors exactly $M/{p^m}$ times. Such an array is denoted by $(M, n, p, m)$. Define $W_i=\{\mathbf{x}\in\mathbb{F}_p^n:f(\mathbf{x})=i\}$ and $b_i=|W_i|$. Construct an array $B_i$ whose rows are elements of $W_i$. Then $B_i$ is a $b_i\times n$ array for $0\leq i\leq p-1$.
\begin{fact}
(\cite{Gopalakrishnan1995}) $f$ is $m$th-order correlation-immune $\Longleftrightarrow$ $B_i$ is a $(b_i,n,p,m)$ orthogonal array for every $i$ with $0\leq i\leq p-1$.
\end{fact}

Fourier spectrum characterization for the correlation-immune Boolean functions was proposed by Wang and Gong in \cite{WangDiscrete}. We shall generalize these results to functions from $\mathbb{F}^n_p$ to $\mathbb{F}_p$ for prime $p$.

\section{Fourier Spectrum Characterization for Correlation-Immune Functions}\label{section:s4}

Let $\pi$ be a permutation of symbols $\{1,2,\cdots ,n\}$, $f_{\pi}=f(x_{\pi(1)},x_{\pi(2)},\cdots ,x_{\pi(n)})$ a function obtained by permuting the variables in $f(x_1,x_2,\cdots ,x_n)$, and $F_{\pi}(z)$ the polynomial associated with the function $f_{\pi}$.

For any integer $d$, let $\Phi_d(z)$ denote the $d$th {\em cyclotomic polynomial} \cite{mceliece1987finite}. Then $\Phi_d(z)$, a monic polynomial with integer coefficients, is the minimal polynomial over the rational field of any primitive $d$th-root of unity. The main result of the paper is given as follows.

\begin{theorem}\label{theorem:t1}
Let $f(x_1,x_2,\cdots ,x_n)$ be a function from $\mathbb{F}^n_p$ to $\mathbb{F}_p$ for $p$ prime. Then $f$ is $m$th-order correlation-immune if and only if $$\Phi_{p^m}(z)|F_{\pi}(z),$$ for all permutation $\pi$.
\end{theorem}

Since $\Phi_{p^m}(z)$ is  the minimal polynomial of $\xi^{-p^{n-m}}$ with respect to polynomial ring with rational coefficients, where $\xi=exp(\frac{2\pi\sqrt{-1}}{p^n})$ in Definition 2,  $\Phi_{p^m}(z)|F_{\pi}(z)$ if and only if $F_{\pi}(\xi^{-p^{n-m}})=0$. Recall the definition of DFT of the function, we have $\mathcal{F}_{f_{\pi}}(p^{n-m})=F_{\pi}(z=\xi^{-p^{n-m}})$. Fourier spectrum characterization of the correlation-immune functions is obtained immediately.

\begin{corollary}
Let $f(x_1,x_2,\cdots ,x_n)$ be a function from $\mathbb{F}^n_p$ to $\mathbb{F}_p$  for $p$ prime. Then $f$ is $m$th-order  correlation-immune if and only if $$\mathcal{F}_{f_{\pi}}(p^{n-m})=0,$$ for all permutation $\pi$.
\end{corollary}

Before giving a proof for Theorem \ref{theorem:t1}, we study the cyclotomic polynomial $\Phi_{p^m}(z)$ first.
It is obvious from the definition of cyclotomic polynomial that
\[\Phi_{p^m}(z)=\prod \{(z-\xi^{j}):0\leq j\leq p^n-1,\mbox{gcd}(j,p^n)=p^{n-m}\},\]
where gcd denotes the great common divisor. On the other hand, since
$$z^N-1=\prod_{d|N}\Phi_d(z),$$
we have
$$z^{p^m}-1=\prod_{j=0}^m\Phi_{p^j}(z)=\Phi_{p^m}(z)\Phi_{p^{m-1}}(z)\cdots\Phi_p(z)\Phi_1(z),$$
and $$z^{p^{m-1}}-1=\prod_{j=0}^{m-1}\Phi_{p^j}(z)=\Phi_{p^{m-1}}(z)\Phi_{p^{m-2}}(z)\cdots \Phi_p(z)\Phi_1(z).$$
So we conclude that
\begin{equation}\label{cyclotomic}
\Phi_{p^m}(z)=\frac{z^{p^m}-1} {z^{p^{m-1}}-1} =\sum_{j=0}^{p-1}(z^{p^{m-1}})^j.
\end{equation}

For ease of illustration, we first consider permutation $\pi$ to be identity, and describe the connection between $\Phi_{p^m}(z)|F(z)$ and probabilistic expression.

\begin{lemma}\label{lemma}
Let $f(x_1,x_2,\cdots ,x_n)$ be a function from $\mathbb{F}^n_p$ to $\mathbb{F}_p$. Then $\Phi_{p^m}(z)|F(z)$ if and only if
$$X_m\rightarrow X_1, X_2, \cdots, X_{m-1}\rightarrow f(X_1,X_2,\cdots ,X_n)$$
is a markov chain, or alternatively, for $\forall t\in \mathbb{F}_p$,
$$P_r \left(f(\mathbf{x})=t|x_1,\cdots ,x_{m-1},x_m \right)=P_r \left(f(\mathbf{x})=t|x_1,\cdots ,x_{m-1} \right).$$
\end{lemma}

\begin{proof}
Since
\begin{equation*}
F(z)=\sum^{p^n-1}_{k=0}\omega^{f(k)}z^k=\sum_\mathbf{x}\omega^{f(\mathbf{x})}\prod^n_{i=1}(z^{p^{i-1}})^{x_i},
\end{equation*}
we have
$$ \Phi_{p^m}(z)|F(z)\Longleftrightarrow F(z)\equiv 0~(\mbox{mod}\ \Phi_{p^m}(z))\Longleftrightarrow \sum_{\mathbf{x}}\omega^{f(\mathbf{x})}\prod^n_{i=1}(z^{p^{i-1}})^{x_i}\equiv 0 ~(\mbox{mod}\  \Phi_{p^m}(z)).$$
From the definition of the cyclotomic polynomial, we know
$$\Phi_{p^m}(z)|z^{p^i}-1, \mbox{for} \ \forall i\geq m, $$
so
\begin{equation}\label{middle}
\Phi_{p^m}(z)|F(z)\Longleftrightarrow \sum_{\mathbf{x}}\omega^{f(\mathbf{x})}\prod^m_{i=1}(z^{p^{i-1}})^{x_i}\equiv 0 ~(\mbox{mod}\  \Phi_{p^m}(z)).
\end{equation}
 Form the formula of the cyclotomic polynomial in (\ref{cyclotomic}), we have
$$(z^{p^{m-1}})^{p-1} \equiv -\sum_{j=0}^{p-2}(z^{p^{m-1}})^j ~(\mbox{mod}\  \Phi_{p^m}(z)).$$
Then the summation in (\ref{middle}) can be divided into two parts, where the first part is for $x_m\neq p-1$ and the second part is for $x_m=p-1$. Hence $\Phi_{p^m}(z)|F(z)$ is equivalent to
\begin{multline}
   \sum_{x_1,\cdots ,x_{m-1}} \sum_{x_m=0}^{p-2} \sum_{x_{m+1},\cdots ,x_n} \omega^{f(\mathbf{x})} \prod_{i=1}^m(z^{p^{i-1}})^{x_i}\\
   +\sum_{x_1,\cdots ,x_{m-1}} \sum_{x_{m+1},\cdots ,x_n} \omega^{f(\mathbf{x})_{|x_m=p-1}}\prod_{i=1}^{m-1}(z^{p^{i-1}})^{x_i}
   \left(-\sum_{j=0}^{p-2}(z^{p^{m-1}})^j \right)=0 ~(\mbox{mod}\  \Phi_{p^m}(z)). \nonumber
\end{multline}
Combining like terms about $z$, the above condition is equivalent to
\[\sum_{j=0}^{p-2}\sum_{x_1,\cdots ,x_{m-1}}\left(\sum_{x_{m+1},\cdots ,x_n} \left(\omega^{f(\mathbf{x})_{|x_m=j}}- \omega^{f(\mathbf{x})_{|x_m=p-1}}\right)\right)z^{jp^{m-1}+\sum_{i=1}^{m-1}p^{i-1}x_i}=0~(\mbox{mod}\  \Phi_{p^m}(z)). \]
Since  $jp^{m-1}+\sum_{i=1}^{m-1}p^{i-1}x_i$, the degree of the item $z^{jp^{m-1}+\sum_{i=1}^{m-1}p^{i-1}x_i}$,  must be less than Euler function $\varphi(p^m)=p^m-p^{m-1}$, which is the degree of $\Phi_{p^m}(z)$, for $j\leq p-2$, it follows that, for $0< x_1,x_2,\cdots ,x_{m-1}\leq p-1$, and $0<j\leq p-2$,
\[\sum_{x_{m+1},\cdots ,x_n} \left(\omega^{f(\mathbf{x})_{|x_m=j,x_1,\cdots ,x_{m-1}}}- \omega^{f(\mathbf{x})_{|x_m={p-1},x_1,\cdots ,x_{m-1}}}\right)=0,\]
which is equivalent to
\[ \sum_{t=0}^{p-1} \left(\# \{\mathbf{x}: f(\mathbf{x})_{|x_m=j,x_1,\cdots ,x_{m-1}}=t \}-\#
    \{\mathbf{x}: f(\mathbf{x})_{|x_m={p-1},x_1,\cdots ,x_{m-1}}=t \}\right)\omega^t=0,\]
where $\# \{ \cdot \}$ denotes the number of elements in the collection $ \{ \cdot \}$.

Recall that the minimal polynomial of $\omega$ is $\Phi_{p}(z)=1+z+z^2+\cdots +z^{p-1}$ with respect to integer polynomial ring, any polynomial $g(z)$ in integer polynomial ring with $g(\omega)=0$ must be a multiple of $\Phi_{p}(z)$. Therefore,
\begin{equation}\label{equation:eq1}
   \# \{\mathbf{x}: f(\mathbf{x})_{|x_m=j,x_1,\cdots ,x_{m-1}}=t \}-\# \{\mathbf{x}: f(\mathbf{x})_{|x_m={p-1},x_1,\cdots ,x_{m-1}}=t\}=c
\end{equation}
 for all $t$, where $c$ is a constant. Since
$$\sum_{t=0}^{p-1} \# \{\mathbf{x}: f(\mathbf{x})_{|x_m=j,x_1,\cdots ,x_{m-1}}=t\}=p^{n-m}$$
for $0<j\leq p-2$, and
  $$\sum_{t=0}^{p-1} \# \{\mathbf{x}: f(\mathbf{x})_{|x_m={p-1},x_1,\cdots ,x_{m-1}}=t\}=p^{n-m}, $$
we obtain that
\begin{equation}\label{equation:eq2}
\sum_{t=0}^{p-1}\left(\# \{\mathbf{x}: f(\mathbf{x})_{|x_m=j,x_1,\cdots ,x_{m-1}}=t\}
-\# \{\mathbf{x}: f(\mathbf{x}_{|x_m={p-1},x_1,\cdots ,x_{m-1}}=t\} \right)=0.
\end{equation}
It follows from (\ref{equation:eq1}) and (\ref{equation:eq2}) that $$\# \{\mathbf{x}: f(\mathbf{x})_{|x_m=j,x_1,\cdots ,x_{m-1}}=t\}-\# \{\mathbf{x}: f(\mathbf{x}_{|x_m={p-1},x_1,\cdots ,x_{m-1}}=t\}=0.$$
In other words,
$$P_r \left(f(\mathbf{x})=t|x_m=j,x_1,\cdots ,x_{m-1}\right)=P_r \left(f(\mathbf{x})=t|x_m={p-1},x_1,\cdots ,x_{m-1} \right)$$
for $\forall t$ and $\forall j$, i.e.,
$$P_r \left(f(\mathbf{x})=t|x_1,\cdots ,x_{m-1},x_m \right)=P_r \left(f(\mathbf{x})=t|x_1,\cdots ,x_{m-1}\right),$$
which complete the proof.
\end{proof}

We now prove Theorem \ref{theorem:t1} by applying permutation $\pi$ and Lemma \ref{lemma}.

\begin{proof}
From Lemma \ref{lemma}, we know that $\Phi_{p^m}(z)|F(z)$ is equivalent to
$$P_r \left(f(\mathbf{x})=t|x_m={j_1},x_1,\cdots ,x_{m-1}\right)=P_r \left(f(\mathbf{x})=t|x_m={j_2},x_1,\cdots ,x_{m-1} \right).$$
For $1\leq s\leq m-1$, $\Phi_{p^m}(z)|F_{\pi}(z)$ for all $\pi=(s,m)$ is equivalent to that $P_r \left( f(\mathbf{x})=t\right)$ does not depend on the values of $x_1, x_2, \cdots x_m$, i.e,
\[P_r \left( f(\mathbf{x})=t|x_1, x_2 \cdots x_m\right) =P_r \left(f(\mathbf{x})=t\right).\]
Then consider all the permutation $\pi$, we obtain
$$P_r \left( f(\mathbf{x})=t|x_{\pi(1)}, x_{\pi(2)}\cdots,x_{\pi(m)} \right)=P_r \left( f(\mathbf{x})=t \right),$$
which is exactly the definition of the $m$th-order correlation-immune function.
\end{proof}

\begin{definition}
A function $f$ is called a {\em symmetric function} if permuting its variables $(x_1,x_2,\cdots ,x_n)$ leads to itself.
\end{definition}
For symmetric function $f$, since $f=f_{\pi}$ for all permutation $\pi$, the Fourier spectral characterization of the correlation-immunity is much simpler.
\begin{corollary}\label{c1}
Let $f(x_1,x_2,\cdots ,x_n): \mathbb{F}_p^n \rightarrow \mathbb{F}_p$ be a symmetric function. Then $f$ is $m$th-order correlation-immune if and only if $$\mathcal{F}_f(p^{n-m})=0.$$
\end{corollary}

\begin{example}
For $p=3,n=4$, $f(x_1,x_2,x_3,x_4)=x_1x_2x_3+x_1x_2x_4+x_1x_3x_4+x_2x_3x_4+x_1x_2+x_2x_3+x_3x_4+x_1x_3+x_1x_4+x_2x_4$ is a symmetric function.

Calculate the Fourier spectrum of $f(x_1,x_2,x_3,x_4)$ at position $3^{4-1}$. We have
$$\mathcal{F}_f(3^{4-1})=0,$$
so $f$ is a first-order correlation-immune function.

Calculate the Fourier spectrum of $f(x_1,x_2,x_3,x_4)$ at position $3^{4-2}$. We have
$$\mathcal{F}_f(3^{4-2})\neq 0,$$
so $f$ is not a second-order correlation-immune function.

\end{example}

\section{Fourier Spectrum Characterization for  Resilient Function}\label{section:s5}
In cryptographic applications, for avoiding statistical dependence between the plaintext and the ciphertext, the function $f(x_1,x_2,\cdots x_n)$ is always required to be balanced. $f(x_1,x_2,\cdots x_n)$ is said to be $m$-resilient if $f(x_1,x_2,\cdots x_n)$ is balanced by fixing $m$ or fewer variables. From Section \ref{section:s4}, we know that
$f(x_1,x_2,\cdots ,x_n)$ is $m$-resilient function if and only if $\mathcal{F}_f(0)=0$ and $\mathcal{F}_{f_{\pi}}(p^{n-m})=0$ for all permutation $\pi$. Here we show more results on the associated polynomial of function $f(x_1,x_2,\cdots x_n)$.

\begin{theorem}
$f(x_1,x_2,\cdots x_n) : \mathbb{F}_p^n \rightarrow \mathbb{F}_p$ is $m$-resilient function if and only if $$(z^{p^m }-1)|F_{\pi}(z),$$ for all permutation $\pi$.
\end{theorem}
\begin{proof}
For ease of illustration, first we consider permutation $\pi$ to be an identity.
\begin{align}\label{equation:eq3}
     (z^{p^m}-1)|F(z)
     &\Longleftrightarrow F(z)\equiv 0~(\mbox{mod}~z^{p^m}-1) \nonumber\\
     &\Longleftrightarrow \sum_{\mathbf{x}}\omega^{f(\mathbf{x})}\prod^n_{i=1}(z^{p^{i-1}})^{x_i}\equiv 0 ~(\mbox{mod}~z^{p^m}-1) \nonumber\\
     &\Longleftrightarrow \sum_{\mathbf{x}}\omega^{f(\mathbf{x})}\prod^m_{i=1}(z^{p^{i-1}})^{x_i}=0\nonumber\\
     &\Longleftrightarrow\sum_{x_1,\cdots ,x_m}\left(\sum_{x_{m+1},\cdots ,x_n}\omega^{f(\mathbf{x})}\right)z^{\sum^{m}_{i=1}p^{i-1}x_i}=0\nonumber\\
     &\Longleftrightarrow \sum_{x_{m+1},\cdots ,x_n}\omega^{f(\mathbf{x})}=0~\mbox{for}~ \mbox{any}~ \mbox{fixed}~x_1,x_2,\cdots ,x_m.
\end{align}
Similar to proof of Theorem~\ref{theorem:t1}, we obtain
$$\# \{\mathbf{x}: f(\mathbf{x})_{|x_1,\cdots ,x_{m-1}}=t \}$$
all have the same value for $0\leq t\leq p-1$. Then $f(\mathbf{x})$ is balanced for any fixed $x_1,x_2,\cdots ,x_m$. It is obvious that $f(\mathbf{x})$ is balanced for any fixed $x_{\pi(1)},x_{\pi(2)},\cdots ,x_{\pi(m)}$ when we apply permutation $\pi$. Thus, $f(\mathbf{x})$ is $m$-resilient function.
\end{proof}

\section{Conclusions}\label{section:s6}
In this paper, we introduced the discrete Fourier transform over the complex field of the function $f(\mathbf{x}): \mathbb{F}_p^n \rightarrow \mathbb{F}_p$, and obtained the Fourier spectral characterization of correlation-immune and resilient functions. That is,
\begin{enumerate}
  \item $f(\mathbf{x})$ is $m$-order correlation-immune if and only if $\mathcal{F}_{f_{\pi}}(p^{n-m})=0$ for all permutation $\pi$.
  \item If $f(\mathbf{x})$ is a symmetric function, $f(\mathbf{x})$ is $m$-order correlation-immune if and only if $\mathcal{F}_f(p^{n-m})=0$.
  \item $f(\mathbf{x})$ is $m$-resilient if and only if $\mathcal{F}_f(0)=0$ and $\mathcal{F}_{f_{\pi}}(p^{n-m})=0$ for all permutation $\pi$.
  \item The method used in \cite{WangDiscrete} cannot be extended to prove the result for $p>2$. However,  Fourier spectral characterization for Boolean functions in \cite{WangDiscrete} can be regarded as a  special case of the results in this paper for $p=2$.
\end{enumerate}
\bibliographystyle{spmpsci}      
\bibliography{biblio}   


\end{document}